\documentclass[runningheads]{llncs}
\usepackage[utf8]{inputenc}
\usepackage{amssymb}
\usepackage{comment}
\usepackage{graphicx}
\usepackage[ruled,vlined,linesnumbered]{algorithm2e}
\usepackage[font=small]{caption}
\def\QED{\hfill$\Box$\par\vskip1em}

\title{Blockchain in Dynamic Networks}

\author{Rachel Bricker \and Mikhail Nesterenko \and Gokarna Sharma}
\institute{Kent State University, Kent, OH, 44242, USA\\
\email{rbricke2@kent.edu}, \email{mikhail@cs.kent.edu}, and \email{gsharma2@kent.edu}}

\begin{document}
\sloppy
\maketitle
\begin{abstract}
We consider blockchain in dynamic networks. We define the Blockchain Decision Problem. It requires miners that maintain the blockchain to confirm whether a particular block is accepted. We establish the necessary conditions for the existence of a solution. We, however, prove that the solution, even under these necessary conditions is, in general, impossible. We then present two algorithms that solve the Blockchain Decision Problem under either the knowledge of the maximum source pool propagation time or the knowledge of the source pool membership. We evaluate the performance of the two algorithms.  
\keywords{Dynamic Networks \and Blockchain}
\end{abstract}

\section{Introduction}
Blockchain is a means of organizing a decentralized public ledger. The lack of centralized controller potentially makes the blockchain more resilient to network failures and attacks. Blockchain is a popular architecture for a number of applications such as cryptocurrency~\cite{nakamoto,ethereum}, massive Internet-of-Things storage~\cite{iotchain} and electronic voting~\cite{grontas2019blockchain}.

The major problem of maintaining this ledger is for the participants to achieve consensus on its records despite faults or hostile environment. Classic robust consensus algorithms~\cite{pbft,byzantine} use cooperative message exchanges between peer processes to arrive at a joint decision. However, such algorithms require that each process is aware of all the other processes in the network. In a system with flexible membership, such requirement may be excessive. 
An alternative is competitive consensus~\cite{nakamoto} where processes race to have records that they generated added to the blockchain. This competition does not require fixed membership and provides defense against attacks and faults. 

Ordinarily, the network underlying the blockchain is considered to be always connected. 
However, as blockchain finds greater acceptance and new applications, this assumption may no longer be considered as
a given. Instead, the blockchain operation under less reliable communication conditions needs to be examined. 

A dynamic network assumes that a connection between any two processes may appear and disappear at any moment. Therefore, the network connectivity graph changes arbitrarily from one point of the computation to the next. This is the least restrictive network connectivity assumption. This paper studies operation of blockchain in dynamic networks. 

\ \\
\textbf{Related work.} 
%
%
An area related to dynamic networks is population protocols, where passive agents do not control their movement but may exchange information as they encounter each other. See Michail et al.~\cite{michail2011new} for an introduction to the topic.
A system with arbitrary link failures was considered by Santoro and Widmayer~\cite{santoro1989time}. There are several papers that explore the model of link failures in greater detail~\cite{afek2013asynchrony,charron2009heard,coulouma2015characterization}.

The network that dynamically changes in an arbitrary manner, possibly to the detriment of the problem to be solved, was first formally studied by O'Dell and Wattenhofer~\cite{o2005information}. This topic is explored in Kuhn et al.~\cite{kuhn2010distributed}.
Several studies~\cite{bonomi2018reliable,guerraoui2021dynamic} investigate reliable broadcast in dynamic networks with Byzantine faults. There is a large body of literature on cooperative consensus in dynamic networks~\cite{biely2012agreement,kuhn2011coordinated,winkler2019consensus}. In particular, Winkler et al.~\cite{winkler2019consensus} explored the concept of an eventually stably communicating root component necessary for consensus. This is similar to the concept of source communication pool that we introduce in this paper. 

There are some applied studies~\cite{cong2021dtnb,hu2019delay} that consider the operation of blockchain that tolerates extensive delays or temporary disconnections. Hood et al.~\cite{hood2021partitionable} explored in detail the blockchain operation under network partitioning.  

However, to the best of our knowledge, this paper is the first to study blockchain in dynamic networks.  


\ \\
\textbf{Paper organization and contribution.} In Section~\ref{secNotation}, we introduce the notation and state the Blockchain Decision Problem for dynamic networks: every network miner needs to confirm the acceptance of each block. In Section~\ref{secDecisive}, we establish the conditions for blockchain and the dynamic network so that the problem is at all solvable: there needs to be a single source pool of continuously interacting miners that propagate the blocks they generate to the rest of the network and none of the other miners may generate infinitely many blocks and propagate them back to the source pool.  In Section~\ref{secImpossibility}, we prove that in general, even if these conditions are met, the problem is impossible to solve. Intuitively, miners may not determine when these outside blocks stop coming. In Section~\ref{secSolutions}, we present two algorithms that solve the problem with restrictions: \emph{KPT} -- if maximum message propagation time is known to all miners, \emph{KSM} -- if source pool membership is known to all miners. We evaluate the performance of the two algorithms in Section~\ref{secPerformance}. We conclude the paper by Section~\ref{secEnd}.

\vspace{-2mm}

\section{Notation, Definitions and Problem Statement}
\label{secNotation}
\textbf{Network.} A network  $N$ consists of a fixed number of processes or \emph{miners}. Each miner has a unique identifier which may or may not be known to the other miners in the beginning. The network computation proceeds in synchronous rounds. 
Miners communicate via message passing over uni-directional links connecting the sender miner $m_s$ and receiver miner $m_r$. This is denoted as $m_s \rightarrow m_r$.
The network is dynamic as links may appear or disappear. More specifically, at the beginning of each round $i$, the receiver miner receives all messages sent to it during the previous round, then carries out calculations and submits messages over the links that exist in round $i$ to be received in the next round.
A \emph{computation} is a, possibly infinite, sequence of such rounds. 

To simplify the presentation, we first assume that the communication is instantaneous. That is, all the information sent over the link is received in the same round. We also assume that arbitrary amount of information may be communicated in one message. We relax these assumptions later in the paper. 

A \emph{journey} in a computation is a sequence of miners and communication links $m_1 \rightarrow \cdots \rightarrow m_i \rightarrow m_{i+1} \rightarrow m_{i+2} \rightarrow \cdots \rightarrow m_x$ such that each round $i$ of the computation, where link $m_i \rightarrow m_{i+1} $ exists, precedes the round with link $m_{i+1} \rightarrow m_{i+2}$. 
\emph{Journey time} is the number of computation rounds between the first and last link in the journey. Note that journey time may be greater than the total number of links in the journey since it may take more than one round for each subsequent link in the journey to appear. 


%

\begin{figure}[htbp]
\vspace{-4mm}
   \centering
   \includegraphics[width=0.73\columnwidth,angle=-90,scale=0.28]{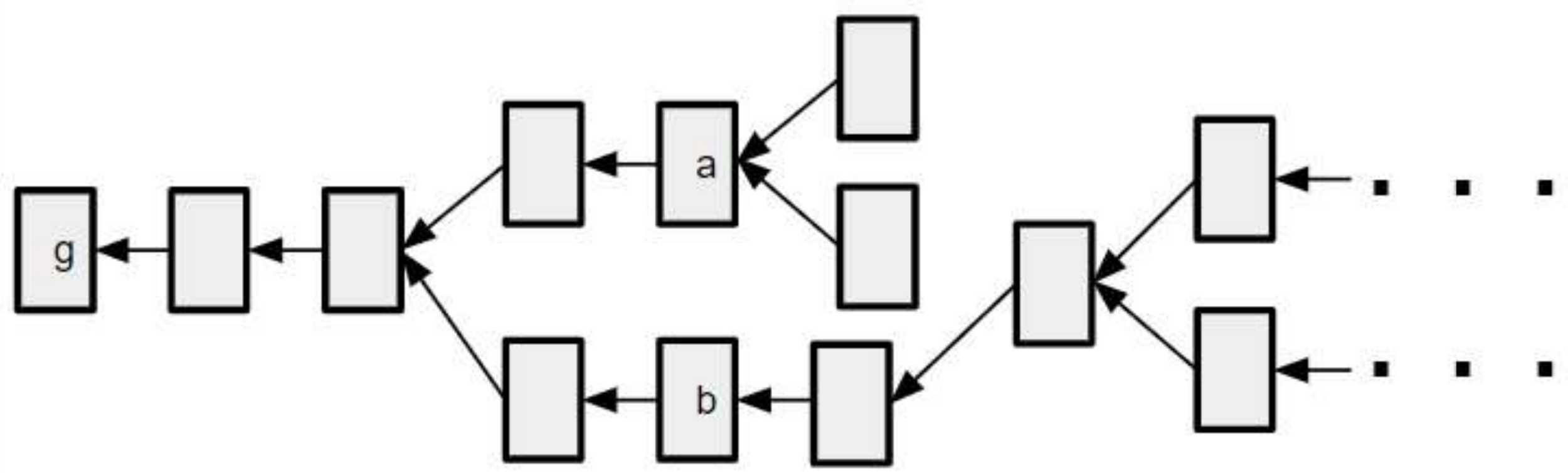}
   \caption{Blockchain notation illustration. Block $g$ is genesis. Block $a$ is rejected because it is an ancestor to only finitely many blocks. Alternatively, $a$ belongs to finite branches only. Block $b$ is accepted since it is an ancestor to infinitely many blocks. That is, $b$ belongs to infinite branches. Block $a$ is a cousin of block of $b$. In an infinite blockchain, $b$ is accepted if all its cousin branches are finite.}
   \label{figBlockchain}
   \vspace{-4mm}
\end{figure}

\noindent
\textbf{Blockchain.} The introduced terms are illustrated in Figure~\ref{figBlockchain}. \emph{Blockchain} is a tree of linked \emph{blocks}.
Each mined block is unique and can be distinguished from the others. Further block contents
is immaterial. A block may be linked to a single \emph{parent block}. A \emph{child} is a block linked to a parent. 
\emph{Genesis} is the root of the tree and the only block without a parent. 
A \emph{leaf} is the block with no children. 
An \emph{ancestor} of a block $b$ is either a parent of $b$ or, recursively, an ancestor of $b$. A \emph{descendant} of a block $b$ is any block whose ancestor is $b$. The \emph{depth} of a block is the number of its ancestors.

A \emph{branch} is the maximal sequence of blocks $b_1, \cdots, b_i, b_{i+1}, \cdots$ such that $b_1$ is the genesis and for each $i$, $b_i$ is the parent of $b_{i+1}$. By this definition, either a branch is infinite or it ends with a leaf. Given a block $b$ that belongs to a branch, all blocks preceding $b$ in this branch are its ancestors and all blocks following $b$ are its descendants. The \emph{length} of a finite branch is the depth of its leaf. The length of an infinite branch is infinite.

A \emph{trunk} of two branches is their longest shared prefix. Thus, the trunk of any two branches is at least the genesis. A branch is a trunk of itself. The blocks of the trunk belong to the branches that share this trunk. Consider a block $b$ that does not belong to the shared trunk of the two branches. \emph{Cousins} of $b$ are the blocks that belong to these branches but neither descendants nor ancestors of $b$. These blocks belong to a \emph{cousin branch}.

Each miner in the network stores all the blockchain blocks known to it. That is, miners maintain local copies of the blockchain. In the beginning of each computation, each miner stores the same genesis.  Due to the haphazard link appearance in a dynamic network, local copies of the blockchain may be out of sync. 

In an arbitrary round, a miner $m$ may generate or \emph{mine} a new block $b$ linked to the longest branch of the local copy of the blockchain. If $m$ has several branches of the same length, the new block may be mined on any one of them. Multiple processes may mine blocks in the same round. Once linked, the sender sends its entire copy of the blockchain to the receiver. We discuss how to limit the amount of transmitted information later in the paper.  
%
%
By this operation, the number of children for any block, i.e. the arity of the blockchain, is at most $|N|$. 

We place few assumptions on the relationship between the relative speed of communication and block mining. However, we assume the following fairness: throughout the computation, a miner either receives infinitely many new blocks or mines infinite many blocks itself. A miner subject to this assumption is a \emph{fair miner}, it is an \emph{unfair miner} otherwise.

\ \\
\textbf{The Blockchain Decision Problem.} A block is \emph{accepted} if it is the ancestor of all but finitely many blocks. A block is \emph{rejected} if it is the ancestor of finitely many blocks.

In the attempt to agree on the common state of the blockchain, each miner decides whether the block is accepted by outputting a \emph{confirm} decision. The decision about block rejection is implied and is not required. To arrive at this decision, the miners may store and exchange arbitrary information. We use the term computation for block mining and blockchain maintenance as well as for the operation of the algorithm that allows the miners to output decision about the blocks of this blockchain. We formulate the  decision problem as follows. 

\begin{definition}[The Blockchain Decision Problem \emph{BDP}]
A solution to the Blockchain Decision Problem satisfies the following properties:
\vspace{-2mm}
\begin{description}
 \item[Decision:] each miner eventually confirms every accepted block; 
 \item[Confirmation Validity:] each miner confirms only accepted blocks.
\end{description}
\end{definition}

\vspace{-2mm}
\section{Decisive Computations}
\label{secDecisive}
\textbf{Globally decisive computations.} A computation is \emph{globally decisive} if every block of its blockchain is uniquely categorized: either accepted or rejected but not both at once.

\begin{lemma}\label{lemSingleInfBranch}
The blockchain of a globally decisive computation has exactly one infinite branch. 
\end{lemma}

To put the lemma another way: in a decisive computation, all branches except for one are finite.

\begin{proof}
The blockchains that do not conform to the conditions of the lemma either have no infinite branches or have more than one.
If a blockchain does not have infinite branches at all, then it has a finite number of blocks. In this case, every block $b$ is the ancestor of finitely many blocks. That is, $b$ is rejected. However, $b$ is also an ancestor of all but finitely many blocks. That is, $b$ is also simultaneously accepted. In a globally decisive computation, a block may be either accepted or rejected but not both. 

Let us consider the second case of a blockchain not conforming to the conditions of the lemma: it has multiple infinite branches. Let block $b$ belong to one such branch but not to the shared trunk of all the branches. Since $b$ belongs to an infinite branch, it is an ancestor to an infinite number of blocks. Therefore, $b$ is not rejected. However, there are infinite number of blocks in the infinite cousin branches, i.e. the branches to which $b$ does not belong. That is, $b$ is not an ancestor to an infinite number of blocks. Hence, $b$ is not accepted either. 

That is, a blockchain with unique categorization of acceptance of rejection has exactly one infinite branch. A computation must have such a blockchain to be decisive. The lemma follows. 
\QED\end{proof}
\vspace{-2mm}
\ \\
\textbf{Mining pools.} In a certain computation, a \emph{mining pool} $M$ is a maximal set of miners such that each miner $m \in M$ has an infinite number of journeys to every other miner in $M$. That is, each miner in a pool is reachable from every other miner in this pool infinitely often. If, for some miner $m$, there are no other miners that are mutually reachable infinitely often, then $m$ forms a pool by itself. 



A pool graph $\cal PG$ for a computation $C$ is a static directed graph formed as follows. Each node in $\cal PG$ corresponds to a mining pool in $C$. An edge from node $P_1 \in \cal PG$ to node $P_2 \in \cal PG$ exists if there is an infinite number of edges from miners of pool $P_1$ to the miners of the pool $P_2$. 

Let us observe that any pool graph $\cal PG$ is a DAG. 
Indeed, if there is cycle in $\cal PG$, then any miner $m_1$ has an infinite number of journeys to any other miner $m_2$ in this cycle. Since mining pools are maximal, these miners belong to the same pool. 
If $\cal PG$ has a path from pool $P_1$ to pool $P_2$, then any miner $m_1 \in P_1$ has an infinite number of journeys to any miner $m_2 \in P_2$. If it does not, then the number of journeys between $P_1$ and $P_2$ is finite.

A node in a  static graph is a source if it has no incoming edges. Since a DAG has no cycles, it has at least one source. A \emph{source pool} is a pool that corresponds to a source in $\cal PG$. 
An infinite branch \emph{belongs} to a pool if it contains a suffix of blocks where every block is mined by a member of this pool. 

\begin{lemma}\label{lemSinglePool}
If the blockchain of a computation contains an infinite branch, this branch belongs to a single pool.
\end{lemma}
\begin{proof} 
Assume that there is a computation $C$ whose blockchain has an infinite branch $BR$ that does not belong to a single pool. That is, branch $BR$ contains infinitely many blocks mined by miners in at least two separate pools $P_1$ and $P_2$. The pool graph $\cal PG$ of $C$ contains no cycles. That means that there is a path from one pool to the other but not back. Suppose, without loss of generality, that there is no path from $P_1$ to $P_2$. This means that there is a finite number of journeys from miners of $P_1$ to $P_2$ in $C$. 
Let round $r$ be the last round of $C$ that contains the journey from $P_1$ to $P_2$. 
However, there are infinitely many blocks in $BR$ that are mined by miners in $P_1$ and in $P_2$. Consider two blocks $b_1$ and $b_2$ of $BR$ mined after round $r$ such that $b_1$ is mined by miner $m_1 \in P_1$ and $b_2$ by miner $m_2 \in P_2$. Moreover, $b_1$ is the ancestor of $b_2$. If this is the case, there is a journey from $m_1$ to $m_2$ in $C$. However, we assumed that there are no such journeys after round $r$ in $C$. That is, our assumption is incorrect and $C$ does not exist. This proves the lemma.
\QED\end{proof}

\vspace{-2mm}
\ \\ 
\textbf{Locally decisive computations.}
A computation is \emph{locally decisive} if it is globally decisive and each miner receives every accepted block.

\begin{lemma}\label{lemInfSource}
In a locally decisive computation, infinite branches belong to a source pool. 
\end{lemma}
\begin{proof}
Assume that there is a locally decisive computation $C$ whose blockchain contains an infinite branch $BR$ that does not belong to source pools of $C$. According to Lemma~\ref{lemSinglePool}, $BR$ belongs to some pool $P$. Since the pool graph of $C$ is a DAG, it must have a source pool $SP$. A source pool has a finite number of journeys from the miners outside itself. 

Computation $C$ is locally decisive. This means that all miners, including the miners in $SP$, receive all blocks in $BR$.  Yet, $BR$ is infinite. This means that there are infinitely many journeys from miners in $P$ to the miners in $SP$. This means, contrary to our initial assumption, that $SP$ is not a source pool.
\QED\end{proof}

\begin{lemma}\label{lemSingleSource}
In a locally decisive computation, there is a single source pool. 
\end{lemma}
\begin{proof}
Assume the opposite: there is a locally decisive computation $C$ with at least two source pools: $SP_1$ and $SP_2$. Since a locally decisive computation is also a globally decisive computation, according to Lemma~\ref{lemSingleInfBranch}, $C$ contains a single infinite branch. According to Lemma~\ref{lemSinglePool}, this branch belongs to a single pool and, according to Lemma~\ref{lemInfSource}, this pool is a source.  That is, the infinite branch belongs to either $SP_1$ or $SP_2$. Let it be $SP_1$. This means that miners of $SP_1$ mine infinitely many blocks that belong to the infinite branch. Since $C$ is globally decisive, these blocks are accepted. 

However, $SP_2$ is also a source, this means that it has a finite number of journeys from miners outside itself. Yet, since $C$ is locally decisive, the miners in $SP_2$ need to receive the infinite number of blocks mined in $SP_1$. That is, there are infinite number of journeys from the miners of $SP_1$ to the miners of $SP_2$. That is, $SP_2$ is not a source. 
\QED\end{proof}

\noindent
The following theorem summarizes the results proven in Lemmas~\ref{lemSingleInfBranch}, \ref{lemSinglePool}, \ref{lemInfSource} and \ref{lemSingleSource}.

\begin{theorem}\label{thrmDecisive}
If a computation is globally and locally decisive, then it has exactly one infinite branch and one source pool. Moreover, this infinite branch belongs to this source pool.
\end{theorem}

\section{Impossibility}
\label{secImpossibility}
\vspace{-3mm}
In a solution to the Blockchain Decision Problem, every miner is required to confirm each accepted block. Theorem~\ref{thrmDecisive} states necessary conditions for the possibility of the solution. Yet, even if these conditions are satisfied, a miner may make a mistake. Indeed, assume a miner $m$ determines that a certain block $b$ belongs to the longest branches of all processes in the source pool. Miner $m$ confirms it.
Yet, a non-source pool miner may later mine a longer cousin branch to $b$, communicate it to the source pool forcing rejection of $b$. This makes $m$'s confirmation incorrect. Even though, by definition of the source pool, such links from the outside happen only finitely many times, the time they stop is not predictable. This makes the solution, in the general case, impossible.
The below theorem formalizes this intuition.

\begin{theorem}\label{thrmNoSolution}
There does not exist a solution to the Blockchain Decision Problem even for globally and locally decisive computations.
\end{theorem}
\begin{proof} 
Assume there is an algorithm $A$ that solves \emph{BDP} for globally and locally decisive computations. Consider a globally and locally decisive computation $C_x$ which contains mining pools $P_1$ and $P_2$ such that $P_1$ is the source pool. 

Since $C_x$ is globally and locally decisive, according to Theorem~\ref{thrmDecisive}, it has a single infinite branch, a single source pool and the branch belongs to this source pool. 
This means that there are infinite number of blocks in the infinite branch. All these blocks are accepted. After some round $r_1$ they must be mined in $P_1$. Let block $b$ be one such block. Since the computation is locally decisive, $b$ has to reach miners in $P_2$. The Decision Property of \emph{BDP} requires that all miners eventually confirm accepted blocks. This means that miners of $P_2$ have to eventually confirm $b$. Let $r_2$ be the round where some miner $m_2 \in P_2$  confirms $b$ in $C_x$. 

Consider a computation $C_y$ that has an extra pool $P_3$. Communication  in $C_y$ is as follows. Miners of pool $P_3$ have no links to the outside miners until round $r_2$.  Since the miners of $P_3$ do not influence other miners, we construct $C_y$ such that up to the round $r_2$, the actions of miners of $P_1$ and of $P_2$ are the same as in  $C_x$.  This includes $m_2$ confirming block $b$. We construct the remainder of $C_y$ as follows. Miners of $P_3$ have only outgoing links to miners of $P_1$ and $P_2$ for the remainder of $C_y$. That is, $P_3$ is a source. We construct $C_y$ to be locally and globally decisive. That is, we make $P_3$ its own single infinite branch. 

By construction, miners of $P_1$ never send messages to miners of $P_3$. This means that block $b$ mined in $P_2$ does not reach $P_3$. Hence, $b$ does not belong to the infinite branch. Therefore, $b$ is rejected. However, miner $m_2$ confirms it in $C_x$ and, therefore, in $C_y$. This is contrary to the Confirmation Validity property of \emph{BDP}, which stipulates that miners may confirm only accepted blocks. Thus, despite our initial assumption, algorithm $A$ does not solve the Blockchain Decision Problem. Hence, the theorem.
\QED\end{proof}

\section{Solutions}
\label{secSolutions}
Previously, we considered completely formed infinite blockchain trees. However, to solve the Blockchain Decision Problem, individual miners have to make decisions whether a particular block is accepted or rejected on the basis of a tree that is not yet complete. Moreover, a miner may not be aware of some already mined blocks due to propagation delays. To describe this uncertainty, we introduce additional notation.

A branch $BR$ is \emph{dead} if all miners are mining on cousin branches longer than $BR$. A branch is \emph{live} otherwise. Notice that once a branch is dead, it may not become live. 
Thus, a block belongs to dead branches only, it is rejected. In an infinite computation, a block is accepted if it belongs to all infinite branches. To put another way, a block is accepted if all its cousin branches are dead. The algorithms in this section exploit the miners' ability to detect dead branches to confirm accepted blocks. 

Per Theorem~\ref{thrmNoSolution}, the solution to the Blockchain Decision Problem is impossible if non-source-pool miners are able to send their mined blocks to the source pool. We, therefore, consider the following restriction. 
A mining pool is \emph{initially closed} if its members do not have incoming edges from non-pool members. If source pool is initially closed, to evaluate whether the branch is dead, it is sufficient to consider blocks generated by source pool miners only. 

\ \\
\textbf{Known propagation time.} Let $m$ be an arbitrary miner in the source pool. \emph{Source pool propagation time} $PT$ is the time of the longest journey from $m$ to any other miner in the network. If $PT$ is fixed, it takes at most $PT$ rounds for a message sent by $m$ to reach all miners. If $PT$ is known, the solution to \emph{BDP} seems straightforward as dead branches eventually become shorter than live ones. However, this solution is not immediate since, even with fixed $PT$, the length difference between live branches may be arbitrarily large. Indeed, a miner may mine a number of blocks extending its branch length significantly. However, other miners may subsequently mine on their branches catching up and keeping their branches live.

\begin{algorithm}[!t]
\small
\SetKwData{accept}{accept}
\SetKwData{reject}{reject}

\textbf{Constants:} \\
    $p$ \tcp*[f]{miner identifier} \\
    $PT$ \tcp*[f]{source pool propagation time, integer} \\
\vspace{2mm}
\textbf{Variables:} \\
 $T$ \tcp*[f]{blockchain tree, initially genesis} \\
 $L$ \tcp{set of tuples $\langle b,l\rangle$, where $b\in T$ and $l$ is either \accept or \reject initially $\varnothing$, if $b \in T$ and $b\not\in L$, then $b$ is unlabeled}
\vspace{2mm}
\textbf{Actions:} \\
\If{mined block $b$}{
   add $b$ to $T$
}
\If{available link to miner $q$}{
   send $T$ to $q$ 
}
\If{receive $T_q$ from miner $q$}{
   merge $T$ and $T_q$  
}
\If{exists unlabeled $b_1$ such that for every $BR(b_1)$, there is a cousin bock $b_2$ such
    that $depth(b_2) > BR(b_1)$ for at least $2\cdot PT$ rounds}
    {
    \label{algKPTlineReject}
    add $\langle b_1,  \reject\rangle$ to $L$
    }
\If{exists unlabeled $b$ such that for its every cousin $c$: $\langle c, \reject\rangle \in L$ }{
    \label{algKPTlineAccept}
     add $\langle b, \accept\rangle$ to $L$ \\
     \textbf{confirm} $b$ 
}
\caption{Known Source Pool Propagation Time Algorithm {\em KPT}.}
\label{algKPT}
\end{algorithm}

Instead, to detect a dead branch, the algorithm that solves \emph{BDP}, relies on the branch length difference over a certain period of time.  
We call this algorithm \emph{KPT}. Its code is shown in Figure~\ref{algKPT}. The algorithm operates as follows. Each miner $p$ maintains the local copy of the blockchain tree $T$ and a set of per-block labels $L$ where it stores decisions whether the block is accepted or rejected. If the decision is not reached, the block is unlabeled. 
Once block $b$ is mined, it is added to the tree $T$. If a link to some miner $q$ appears, miner $p$ sends its entire blockchain to $q$.  

The decisions are reached as follows. An unlabeled block $b_1$  is labeled rejected if for its every branch $BR(b_1)$ the following happens. There is a cousin block $b_2$ such that the depth of $b_2$ is greater than the length of this branch $BR(b_1$) for at least $2\cdot PT$ rounds. An unlabeled block $b$ is accepted if all its cousins are rejected. In the latter case, $b$ is confirmed.  

\begin{lemma}\label{lem2PT}
Let, at some round $r$, some miner $m$ observe that there is a block $b$ whose depth is greater than the length of its cousin branch $BR$. If $b$'s depth is still greater than the length of $BR$ at round $r+ 2\cdot PT$, then $BR$ is dead. 
\end{lemma}
\begin{proof}
Assume block $b$ is mined by miner $m_b$ of the source pool in round $r_b$. Since the maximum source pool propagation time $PT$ is fixed, in round $r_b+PT$ every miner receives $b$. If there is a branch $BR$ that is shorter than the depth of $m_b$ in round $r_b+PT$, then every miner $m_c$ that mines on branch $BR$ or shorter branches, switches to a branch that contains $b$ or a longer branch. That is, $BR$ is dead by $r_b + PT$. Therefore, blocks may be mined on $BR$ no longer than round $r_b + PT$. It takes any block at most $PT$ rounds to propagate to all miners in the network. 

Let us consider an arbitrary miner $m$. Since $b$ is mined in round $r_b$, the earliest round  when $m$ may receive $b$ is also $r_b$. 
That is, if some miner $m$ observes that there is a block whose depth is greater than some  branch for $2\cdot PT$ rounds, then this branch is dead. 
\QED\end{proof}

\begin{theorem}\label{thrmKPT}
Known Source Pool Propagation Time Algorithm {\em KPT} solves the Blockchain Decision Problem with initially closed source pool.
\end{theorem}
\begin{proof} 
Let us consider the Confirmation Validity Property of \emph{BDP}. 
According to Lemma~\ref{lem2PT},  If $PT$ is known and if miner $p$ observes that some block is deeper than the height of a branch for longer than $2\cdot PT$ rounds, then this branch is dead. If some block belongs to dead branches only, it is rejected. This is the exact condition under which blocks are labeled rejected in \emph{KPT}, see Line~\ref{algKPTlineReject}.
If all cousins are rejected, the block is accepted. This is how the is block is labeled accepted and confirmed in \emph{KPT}, see Line~\ref{algKPTlineAccept}.
To put another way, \emph{KPT} confirms only accepted blocks which satisfies the Confirmation Validity Property of \emph{BDP}.

Let us now discuss the Decision Property and show that every accepted block is eventually confirmed. Indeed, according to Lemma~\ref{lem2PT}, a miner determines that a branch is dead in at most $2\cdot PT$ rounds. A block is rejected once all branches that it belongs to are dead. That is, a block rejection is determined in this many rounds after the last branch of the block is dead. 

A block is labeled accepted and then confirmed after all its cousins are rejected. To put another way, a block is accepted after at most $2\cdot PT$ rounds of the rejection of the last cousin block. This proves that all accepted blocks are eventually confirmed and \emph{KPT} satisfies the Decision Property of \emph{BDP}.
\QED
\end{proof}

Let us describe a couple of simple enhancements of \emph{KPT}. Since miners never make mistakes in their classification of reject and accept, a miner may send its label set $L$ to help its neighbors make their decisions faster. Also, a miner may determine dead branches quicker if each block is labeled with the round of its mining. In this case, to ascertain that a certain branch $BR$ is dead, it is sufficient to check if there is a cousin block $b_2$  such that $BR$ does not outgrow $b_2$ for $PT$ rounds.

\ \\
\textbf{Known source pool membership.}
Miner \emph{position} in a blockchain tree is the block on which it is currently mining. Note that the depth of a miner's position throughout the computation may only increase. Once it is observed that all source pool miners moved to positions longer than a particular branch, the source pool miners may not mine on this branch. That is, the branch is dead. We state this formally in the following lemma.

\begin{lemma} \label{lemPosition}
If some miner $m$ observes that there is a branch $BR$  such that the depth of the position of every source pool miner is greater than the length of $BR$, then $BR$ is dead.
\end{lemma}

\begin{algorithm}[htb]
\small
\SetKwData{accept}{accept}
\SetKwData{reject}{reject}

\textbf{Constants:} \\
    $p$ \tcp*[f]{miner identifier} \\
    $SM$ \tcp*[f]{set of ids of source pool miners} \\
\vspace{2mm}
\textbf{Variables:} \\
 $T$ \tcp*[f]{blockchain tree, initially genesis} \\
 $P$ \tcp{set of tuples $\langle b, m\rangle$, where $b\in T$ and $m \in SM$, initially $\varnothing$, positions of source pool miners}
 $L$ \tcp{set of tuples $\langle b,l\rangle$, where $b\in T$ and $l$ is either \accept or \reject initially $\varnothing$, if $b \in T$ and $b\not\in L$, then $b$ is unlabeled, block labels}
\vspace{2mm}
\textbf{Actions:} \\
\If{mined block $b$}{
   add $b$ to $T$ \\
   \If{$p \in SM$}{
    update $p$'s entry in $P$ to  $\langle b, p\rangle$
    }
}
\If{available link to miner $q$}{
   send $T, P$ to $q$ 
}
\If{receive $T_q, P_q$ from miner $q$}{
   merge $T$ and $T_q$, merge $P$ and $P_q$ \\
   \If{$p \in SM$}{
    let $b$ be the deepest block in $T$ \\
    update $p$'s entry in $P$ to  $\langle b, p\rangle$
    } 
}
\If{exists unlabeled $b_1$ such that for every $BR(b_1)$, for all $m \in SM$ 
    there exists $\langle b_2, m \rangle \in P$ such that  $depth(b_2) > BR(b_1)$} 
    {
    add $\langle b_1,  \reject\rangle$ to $L$
    }
\If{exists unlabeled $b$ such that all cousins of $b$ are \reject}{
     add $\langle b, \accept\rangle$ to $L$ \\
     \textbf{confirm} $b$ 
}
\caption{Known Source Pool Membership Algorithm {\em KSM}.}
\label{algKSM}
\end{algorithm}

Determining source pool miner positions directly from mined blocks in the blockchain is not always possible: some source pool miner, even if it is fair, may never mine a block if it keeps receiving longer branches. Instead, the below algorithm relies on miners directly reporting their positions. We call this algorithm \emph{KSM}. Its code is shown in Figure~\ref{algKSM}. Similar to \emph{KPT}, it maintains the blockchain tree $T$ and a set of accept/reject labels per each block $L$. Besdies those, \emph{KSM} also maintains set $P$ where it records the positions of all miners in the source pool. Each miner sends its collected positions together with the blockchain along all outgoing links. 
A block is rejected if all its branches are shorter than the known positions of the source pool miners. Note that non-source pool miners may still mine on the dead branches and extend them. However, since the source pool is closed and the source pool miners never see these non-source pool generated blocks, they are never added to the live branches. 
The block labeling is similar to \emph{KPT}. Once the dead branches are determined and the rejected blocks are labeled, the blocks whose cousins are dead are accepted and confirmed. The correctness argument is similar to that of \emph{KPT}. It is stated in Theorem~\ref{thrmKSM}.

\begin{theorem}\label{thrmKSM}
Known Source Pool Membership  Algorithm {\em KSM} solves the Blockchain Decision Problem with initially closed source pool.
\end{theorem}
Observe that in \emph{KSM}, all miners know the source pool membership. Thus, a miner that is not in the source pool knows that all the blocks that it mines are rejected. So this miner may either not mine its own blocks at all or discard them as soon as they are mined. 

\vspace{-3mm}
\section{Extensions and Optimizations}
\vspace{-3mm}
The algorithm presentation and discussion in the previous sections focused on simplicity. However, there are optimizations that can be implemented to make the algorithms more applicable and more generic. We are going to list them here. 

In the previous section, we assumed that the pool is initially closed. However, both algorithms could be modified to operate correctly if there is a known upper bound when the source pool is closed. That is, all miners are aware of the round number after which there are no incoming links for source pool miners from non-source pool miners. 

Also, we assumed that each miner is sending the entire copy of its blockchain. This is unnecessary. First, with no modifications, both algorithms operate correctly even if each miner sends only its longest branch. That is, the branch that it is currently mining on. However, further sending optimization is possible. Observe that the operation of the algorithms hinges on the miners communicating infinitely often. Thus, if a miner keeps track of the blocks it already sent, it is sufficient to send only the oldest, i.e. the deepest unsent block over each link. With this modification, the two algorithms, \emph{KPT} and \emph{KSM}, transmit only $\log N$ bytes in every message. That is, the two algorithms use constant size messages.

We assumed that link communication is instantaneous. However, the algorithms remain correct even if a message in each link is delayed for arbitrary time.   

It is interesting to consider message loss. If there is fair message loss that allows ultimate progress, the two algorithms operate correctly if each miner sends the entire blockchain or the longest branch in every message. However, since the feedback communication between receiver and sender is not guaranteed, we suspect that constant message size algorithm for either known pool membership or known propagation time does not exist.
\vspace{-3mm}
\section{Performance Evaluation}
\label{secPerformance}
\vspace{-2mm}
For our performance evaluation studies, we used QUANTAS abstract simulator~\cite{oglio2022quantitative}. We generated dynamic topologies as follows. The maximum number of potential neighbors $mx \leq |N|-1$ was fixed. Each round, for every miner, the number of actual neighbors was selected uniformly at random from $0$ to $mx$; the neighbor identifiers were also selected randomly. Miners generated blocks at the rate of $2.5\%$.

\begin{figure}[htbp]
\hspace{-3mm}
\begin{tabular}{c@{\hspace{3mm}}c}
\begin{minipage}[t]{0.51\textwidth}
   \includegraphics[width=\columnwidth]{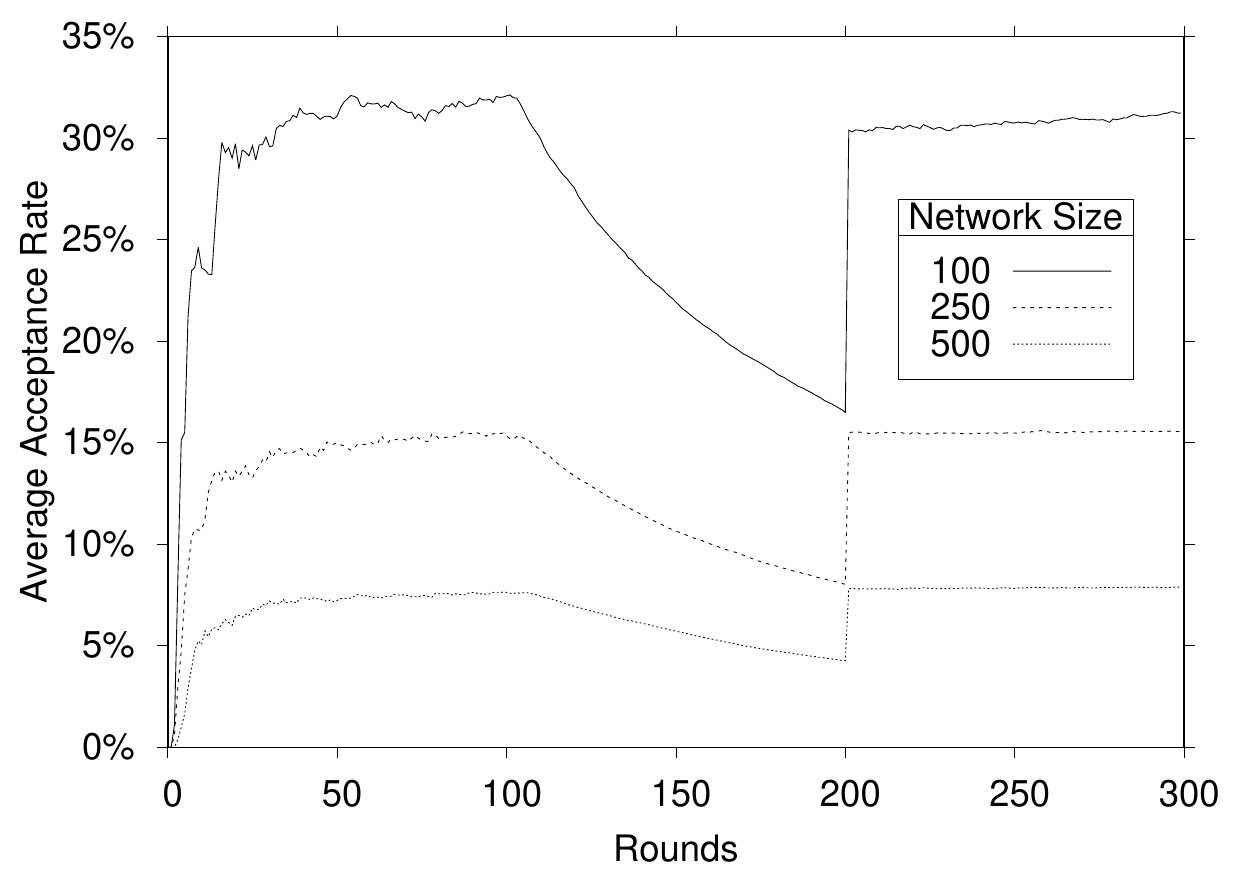}
   \vspace{-4mm}
   \caption{Acceptance rate change as source pool membership is modified: complete network from rounds $0$ to $99$, $25\%$ from $100$ to $199$, and complete network thereafter.}
   \label{figAcceptance}
\end{minipage}
&
\begin{minipage}[t]{0.51\textwidth}
   \includegraphics[width=\columnwidth]{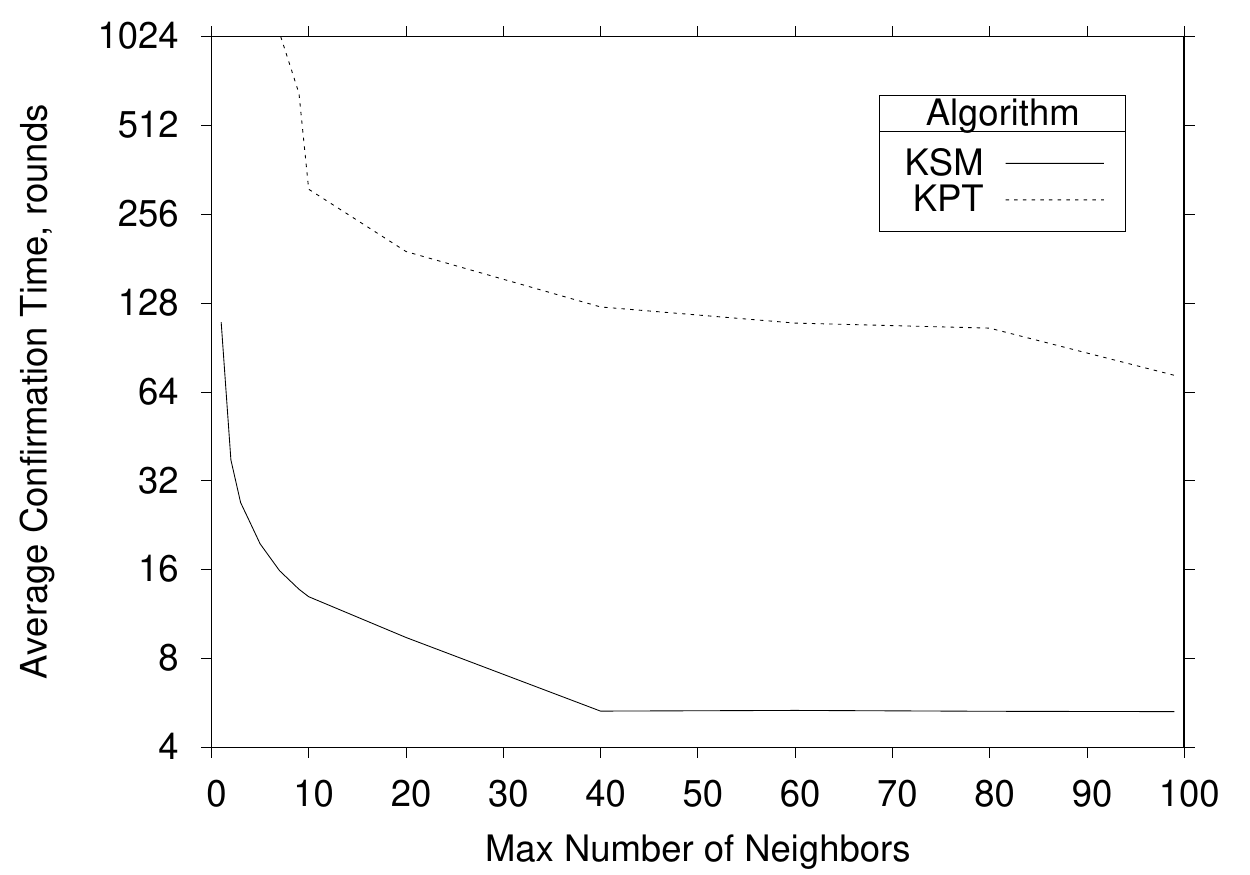}
   \vspace{-4mm}
   \caption{Confirmation time vs. maximum neighborhood size.}
   \label{figMaxNeighbors}
\end{minipage}
\end{tabular}
\vspace{-8mm}
\end{figure}

In the first experiment, we studied the dynamics of block acceptance as the source pool membership changed. The results are shown in Figure~\ref{figAcceptance}. We ran the computations for $300$ rounds. In the first $100$ rounds, the neighbors were selected from the whole network. That is, the complete network was the source pool. In the second $100$ rounds, $\lfloor 25\rfloor\%$ of miners were selected to be the source pool. Specifically, the source pool miners may connect to arbitrary neighbors, i.e. they have no connection restrictions. The remaining miners may connect only to non-source pool miners. In the remaining $100$, the restrictions were lifted and all miners formed the single source pool again. 

In a particular state of the computation, some block is accepted if it is in the longest branch of every miner. That is, every miner is mining on top of this block. 
The \emph{acceptance rate} is the ratio of accepted vs. generated blocks. We ran experiments for the network size of $100$, $250$ and $500$ miners. We did $10$ experiments per network size and averaged our results. 

The results indicate that, as the source pool size is restricted, the block acceptance rate declines. This is due to the source pool neighbors not receiving the blocks from non-source pool neighbors. The acceptance rate sharply rises as the the source pool is enlarged to incorporate all miners and long chains of blocks mined outside the source pool are propagated throughout the network. The acceptance rate is lower in the networks of larger size. Indeed, as more concurrent blocks are generated, fewer of them are accepted. 

In the next experiment, we observed how the neighborhood size affects the time it takes our algorithms to confirm the blocks. We implemented \emph{KSM} and \emph{KPT} and measured their confirmation time. The \emph{confirmation time} for a particular block is the number of rounds from the round when the block was generated till the round when the last miner outputs the confirmation decision. We counted confirmation time for accepted blocks only. We varied the maximum number of neighbors $mx$ and observed average confirmation time for \emph{KSM} and \emph{KPT}. The network size was $100$, the source pool was fixed at $75$ miners.

Algorithm \emph{KPT}, needs maximum propagation time \emph{PT} to be known in advance. 
To determine \emph{PT} we ran preliminary computations. For a fixed $mx$, we computed 
\emph{PT} by running $100$ computations with this $mx$ and computing the longest recorded propagation time. These preliminary computation lengths were set between $10,000$ and $15,000$ rounds. Then, for measurement computations, to collect sufficiently many confirmed blocks, we set computation lengths to $12\cdot PT$ rounds.  We ran $10$ experiments per data point. 

The results are shown in Figure~\ref{figMaxNeighbors}. They indicate that, as the maximum possible number of neighbors increases, the blocks are propagating faster and the confirmation time drops. Perhaps surprisingly, \emph{KSM} performed better because each miner can confirm a block as soon as it receives the data from all the known source pool miners, while, in \emph{KPT}, a miner has to wait for twice the maximum propagation time \emph{PT}. This holds even though we ran preliminary computations to select the shortest possible maximum propagation time \emph{PT}.  

\begin{figure}[htbp]
\vspace{-5mm}
\hspace{-3mm}
\begin{tabular}{c@{\hspace{3mm}}c}
\begin{minipage}[t]{0.51\textwidth}
   \includegraphics[width=\columnwidth]{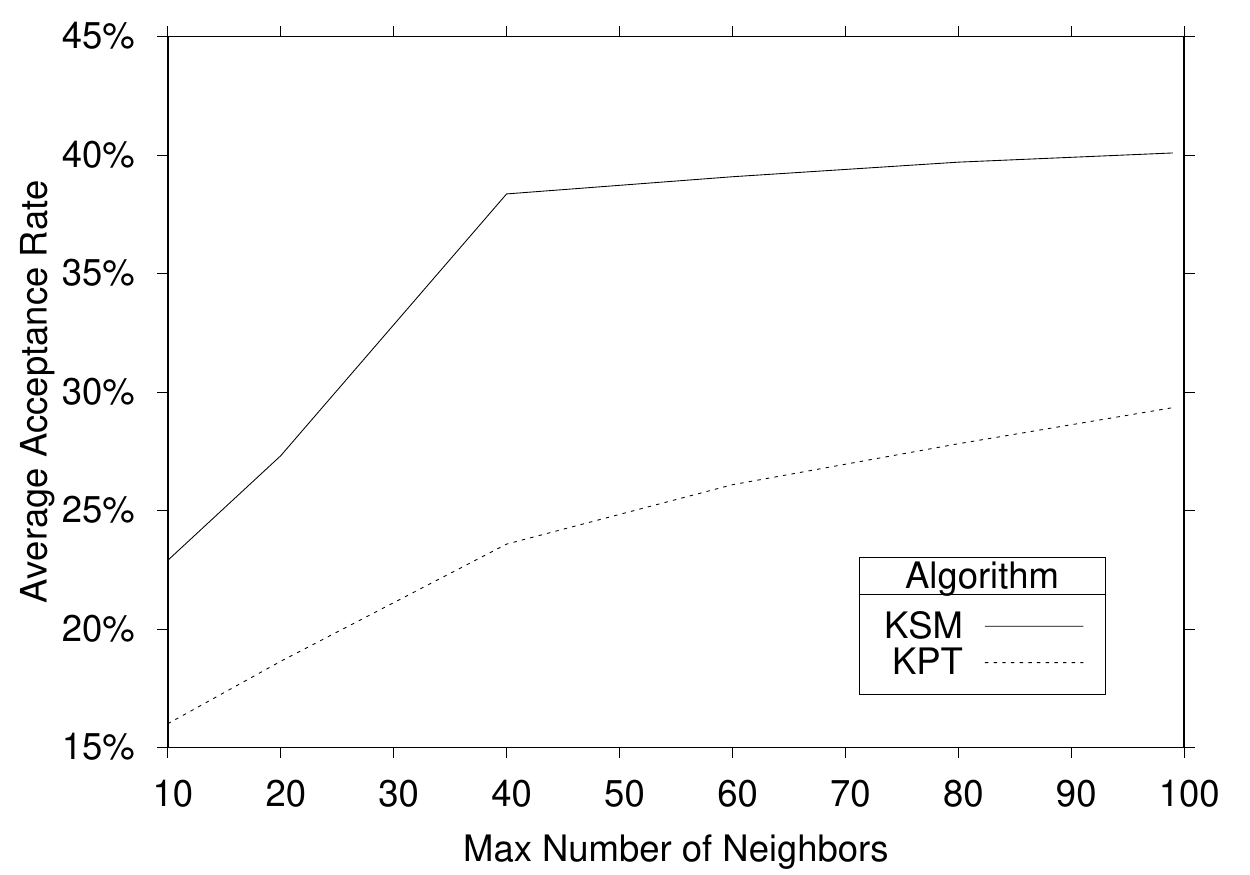}
   \vspace{-4mm}
   \caption{Acceptance Rate vs. maximum neighborhood size.}
    \label{figAlgorithmsAcceptance}
\end{minipage}
&
\begin{minipage}[t]{0.51\textwidth}
   \centering
   \includegraphics[width=\columnwidth]{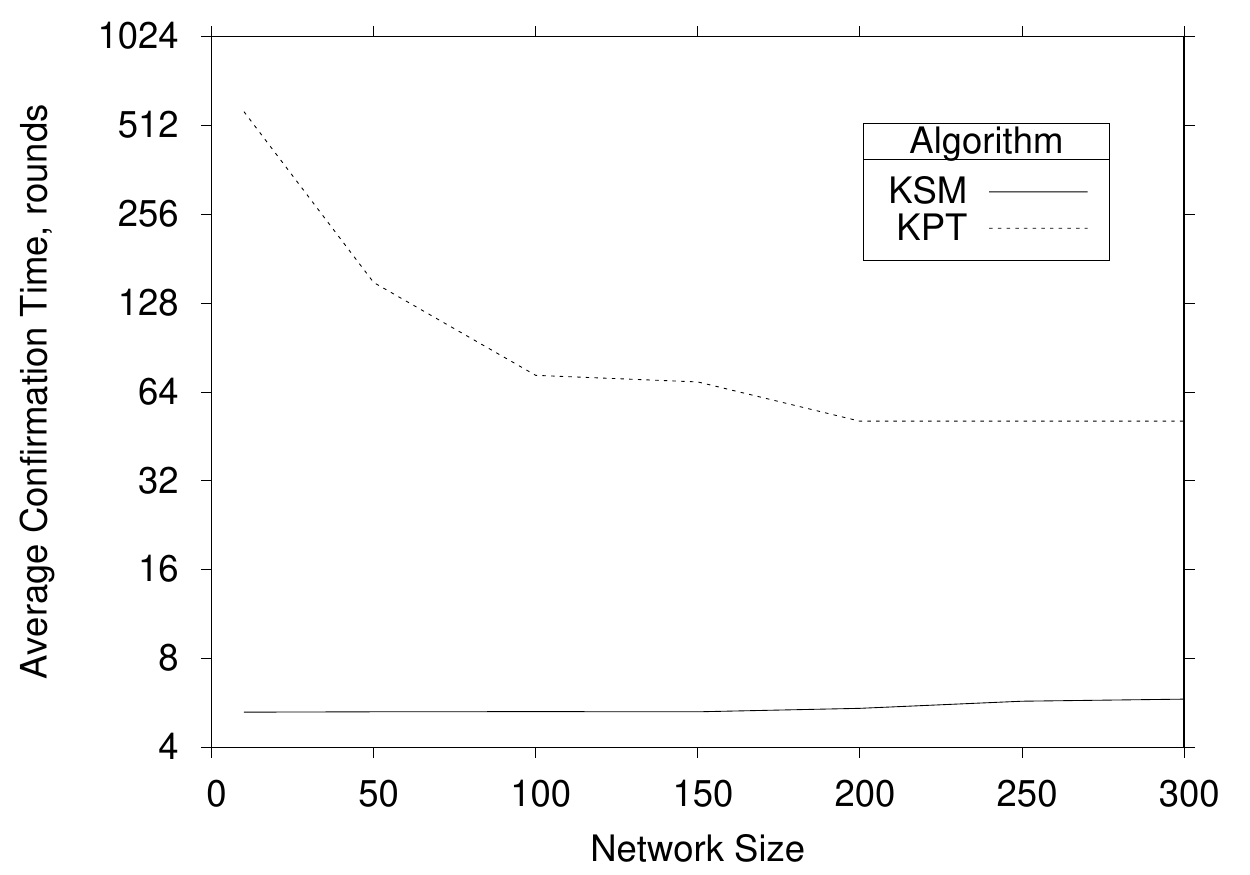}
   \vspace{-4mm}
   \caption{Confirmation time vs. network size.}
   \label{figScale}
\end{minipage}
\end{tabular}
\end{figure}

\vspace{-7mm}
For the same experiment, we computed average acceptance rate. We show the results in Figure~\ref{figAlgorithmsAcceptance}. Algorithm \emph{KSM} has lower confirmation time and, therefore, higher acceptance rate. 

In the final experiment, we observed the performance of the two algorithms as the network scale changes. The number of source pool members is fixed at $\lfloor 75\rfloor\%$ of the network size. The computation lengths were set to $12\cdot PT$ rounds. We ran $10$ computations per data point. The results are shown in Figure~\ref{figScale}. As the network scale increases, $mx$ increases also. This increases the number of potential journeys and decreases $PT$, which, in turn, decreases the confirmation time of \emph{KPT} that depends on $PT$. \emph{KSM} exhibits the opposite dynamics. With larger scale, the number of source pool miners increases also. This makes \emph{KSM} run slightly slower as every miner has to wait to hear from a greater number of source pool miners. 

Our performance evaluation shows that \emph{KSM} outperforms \emph{KPT} under all conditions. Therefore, \emph{KPT} should be considered only when the source pool membership is not available and \emph{KSM} is not implementable.
\vspace{-4mm}
\section{Conclusion}
\label{secEnd}
\vspace{-2mm}
In this paper, we studied the operation of blockchain in dynamic networks. We explored how blockchain behaves at the boundaries of connectivity: where message delay and miner participation is tenuous while connection and communication speeds vary greatly. We believe this contributes to the knowledge of blockchain as a construct and helps engineers to design blockchain for high-stress, uncertain communication environments. 

\vspace{-4mm}

\bibliographystyle{abbrv}
\bibliography{dnblockchain}

\end{document}